\documentclass[reprint,aps,prl,twocolumn,showpacs,superscriptaddress]{revtex4-1}
\usepackage{graphicx}
\usepackage{amsmath}
\usepackage{amsthm}
\usepackage{amsfonts}
\usepackage{amssymb}
\usepackage{xcolor}
\usepackage{hyperref}
\usepackage[normalem]{ulem}

\hypersetup{
	colorlinks,
	linkcolor={blue},
	citecolor={blue},
	urlcolor={blue}
}

\usepackage{cleveref}

\newtheorem{theorem}{Theorem}

\def\cT{\ensuremath{\mathcal{T}}}

\newcommand{\set}[1]{\left\{ #1 \right\}}

\newcommand{\Tr}[1]{\rm Tr\left[#1\right]}

\begin{document}

\graphicspath{{./figures/}}

\title{Contextuality and the Single-Qubit Stabilizer Subtheory}
\date{\today}
\author{Piers Lillystone}
\affiliation{Institute for Quantum Computing and Department of Physics and Astronomy, University of Waterloo, Waterloo, Ontario N2L 3G1, Canada}
\author{Joel J. Wallman}
\affiliation{Institute for Quantum Computing and Department of Applied Mathematics, University of Waterloo, Waterloo, Ontario N2L 3G1, Canada}
\author{Joseph Emerson}
\affiliation{Institute for Quantum Computing and Department of Applied Mathematics, University of Waterloo, Waterloo, Ontario N2L 3G1, Canada}
\affiliation{Canadian Institute for Advanced Research, Toronto, Ontario M5G 1Z8, Canada}

\begin{abstract}

Contextuality is a fundamental non-classical property of quantum theory, which has recently been proven to be a key resource for achieving quantum speed-ups in
some leading models of quantum computation. However, which of the forms of
contextuality, and how much thereof, are required to obtain a speed-up in an
arbitrary model of quantum computation remains unclear. In this paper, we show
that the relation between contextuality and a compuational advantage is more
complicated than previously thought. We achieve this by proving that generalized contextuality is
present even within the simplest subset of quantum operations, the so-called
single-qubit stabilizer theory, which offers no computational advantage and was
previously believed to be completely non-contextual. However, the contextuality
of the single-qubit stabilizer theory can be confined to transformations.
Therefore our result also demonstrates that the commonly considered
prepare-and-measure scenarios (which ignore transformations) do not fully
capture the contextuality of quantum theory.

\end{abstract} 

\maketitle

Contextuality~\cite{Bell1967, Kochen1967, Mermin1990, Peres1991, Mermin1993,
Spekkens2005,Simmons2017}, which includes the better-known concept of Bell non-locality as a
special case, is often regarded as the fundamental non-classical property of
quantum theory. Furthermore, contextuality has emerged as an intriguing
explanation for the power of quantum computation: as contextuality is
required~\cite{Howard2014, Delfosse2015, Bermejo-vega2017, Raussendorf2017} to
achieve an exponential quantum speed-up by injecting magic states into Clifford
circuits~\cite{Bravyi2005}; and also quantifies the computational advantage that
can be obtained~\cite{Veitch2012, Veitch2014, Bermejo-vega2017, Okay2017,
Catani2017} in both the magic-state and measurement-based models of quantum
computation~\cite{Raussendorf2001}. 

These considerations motivate us to understand the scope of phenomena that exhibit contextuality,
with the aim of identifying which features of contextual phenomena enable quantum computational speed-up.
However, one of the primary
obstacles to understanding how contextuality powers a quantum computer is that
the multi-qubit stabilizer subtheory~\footnote{A subtheory of quantum theory is
a closed subset of the operations available within quantum theory for a given
Hilbert space dimension.} exhibits contextuality and
yet can be efficiently simulated on a classical
computer~\cite{Gottesman1998,Aaronson2004}. 
 
There are two leading definitions of contextuality: traditional contextuality ~\cite{Bell1967,
Kochen1967, Mermin1990, Peres1991, Mermin1993} and generalized contextuality ~\cite{Spekkens2005,Simmons2017}.
In this paper, we show that, generalized contextuality~\cite{Spekkens2005} is
present even in the single-qubit stabilizer subtheory of quantum theory, a fact missed by previous work ~\cite{Wallman2012c,
Blasiak2013, Kocia2017}. We further
demonstrate that the contextuality present in the single-qubit stabilizer subtheory can
be confined to only appear in the transformations. This contradicts the common --- and often
implicit --- assumption that an operational theory can be classified as contextual
or non-contextual by only considering the preparations and
measurements~\cite{Harrigan2007, Spekkens2008, Kunjwal2015, Bermejo-vega2017, Krishna2017, Schmid2017}.

\textit{Operational theories and ontological models thereof.---}An operational
theory is non-contextual under a given definition if there exists an
ontological model of the operational theory satisfying a specific property that we describe below.

An operational theory is a mathematical framework for predicting the outcomes of
an experimental procedure, that is, a sequence of preparations, transformations, and
measurements. These elements fully determine the experimental statistics, that
is, the conditional probabilities $\text{Pr}(k|P,T,M)$ of observing the outcome
$k$ when the preparation $P$, transformation $T$, and measurement $M$ are
performed sequentially. In quantum mechanics, the conditional probabilities for
an experiment consisting of preparing a density matrix $\rho$, applying a
completely positive and trace-preserving (CPTP) map $\Phi$, and measuring a
positive-operator-valued measure (POVM) $\{E_k\}$ are
$\text{Pr}(k|\rho,\Phi,E_k) = \Tr{E_k \Phi(\rho)}$ by the Born rule.

We can describe the underlying physical processes that generate the experimental
statistics using the ontological models formalism. Here we follow the treatment
in Ref.~\cite{Leifer2014}. An ontological model is defined by a measurable space
$\Lambda$ of possible physical states, with an associated $\sigma$-algebra
$\Sigma$, and sets of measures or measurable functions over $\Lambda$ are used
to represent preparations, transformations and measurements in the ontological
model. For simplicity, we assume that there exists a measure that dominates all other measures in the model \footnote{That is $\exists \sigma$ such that $\sigma(\Delta) = 0 \Rightarrow \mu(\Delta) = 0, \, \Delta \in \Sigma$ for all $\mu$ in the model.}
(see the Appendix for a proof of the main theorem without this assumption).
This allows us to express an ontological model in terms of probability densities, stochastic matrices, and response functions.

When a system is prepared via some procedure $P$, the physical properties
of the system are probabilistically assigned values, which are completely
encoded by the physical states $\lambda \in \Lambda$. Mathematically, we
associate each preparation procedure $P$ with a probability density over
$\Lambda$, $\mu_P:\Lambda\to[0,1]$, where $\int_\Lambda\mu_P(\lambda)d\lambda=1$
as a system is always in some physical state. That is, the probability that a
physical state, $\lambda\in\Lambda$, was prepared via $P$ is $\mu_P(\lambda)$.

Similarly, when a transformation is applied to a system, the physical properties
of the system dynamically evolve according to some stochastic map. Formally, we
associate each transformation procedure $T$ with a stochastic map
$\Gamma_T:\Lambda\times\Lambda\to [0,1]$, where the conditional probability that
some $\lambda$ is sent to another $\lambda^\prime$ by $T$ is
$\Gamma_T(\lambda^\prime,\lambda)$. As every physical state is mapped to some
physical state by a transformation, $\int_\Lambda
\Gamma_T(\lambda^\prime,\lambda) d\lambda^\prime = 1$ for all $\lambda \in
\Lambda$.

Finally, when a system is measured via some procedure $M$, the probability that
outcome $k$ occurs is specified by the physical state. That is, a measurement
$M$ is equivalent to a set of conditional probability functions
$\{\xi_k^M:\Lambda\to[0,1]\}_k$. As some measurement outcome always occurs,
$\sum_k \xi_k^M(\lambda)=1$ for all $\lambda\in\Lambda$. To correctly reproduce
the experimental statistics of the operational theory, the distributions must
satisfy
\begin{align}\label{TotProb} 
\text{Pr}(k|P, T, M) =  \int_\Lambda 
\xi_k^M(\lambda^\prime) \Gamma_{T}(\lambda^\prime, \lambda)\mu_{P}(\lambda)
d\lambda d\lambda^\prime. 
\end{align}

Ontological models are assumed, often implicitly ~\cite{Gurevich2015}, to be convex-linear, that is, a
probabilistic implementation of a set of operations is represented by the
probabilistic mixture of the corresponding probability densities.

\textit{Generalized Contextuality.---}We now review generalized contextuality.
The (experimental) setting of an operation is the set of other operations that
are performed with the operation during an experiment. Two operations are
operationally equivalent, denoted $\cong$, if they produce the same outcome
statistics in all settings.
\begin{itemize}
\item Two preparations $P$ and $P^\prime$ are equivalent, ($P\cong P^\prime$),
if $\text{Pr}(k|P, T, M) =\text{Pr}(k|P^\prime, T, M)\ \forall T,M$; 
\item Two transformations $T$ and $T^\prime$ are equivalent, ($T\cong
T^\prime$), if $\text{Pr}(k|P, T, M) =\text{Pr}(k|P, T^\prime, M)\ \forall P,M$;
and
\item Two measurement outcomes $k \in M$ and $k \in M^\prime$ are equivalent,
[$(k,M)\cong (k,M^\prime)$], if $\text{Pr}(k|P, T, M) =\text{Pr}(k|P, T,
M^\prime)\ \forall P,T$. 
\end{itemize}
Note that the definition of operational equivalence differs slightly from that
of Ref.~\cite{Spekkens2005} in that we consider operational equivalence of
individual measurement outcomes. However, this definition can be obtained from
that of Ref.~\cite{Spekkens2005} by coarse-graining all measurements into
two-outcome POVMs ~\cite{Leifer2005,Spekkens2014,Kunjwal2015}.

An ontological model is preparation non-contextual (PNC) if operationally
equivalent preparation procedures are represented by the same probability
densities, that is,
\begin{align}
\mu_P = \mu_{P^\prime} &\Leftrightarrow P\cong P^\prime \label{eq:PNC}.
\end{align}
Similarly, an ontological model is transformation non-contextual (TNC)  if 
\begin{align}
\Gamma_T = \Gamma_{T^\prime} &\Leftrightarrow T\cong T^\prime \label{eq:TNC} 
\end{align}
and measurement non-contextual (MNC) if
\begin{align}
\xi_{k,M} = \xi_{k,M^\prime} &\Leftrightarrow (k,M)\cong(k, M^\prime). \label{eq:MNC}
\end{align}
An ontological model is universally non-contextual, in the generalized sense, if
it satisfies \cref{eq:PNC,eq:TNC,eq:MNC}, otherwise it is
contextual~\cite{Spekkens2005}.

Even a single qubit manifests generalized contextuality ~\cite{Spekkens2005}. However, previous
proofs of generalized contextuality for a single qubit have required subtheories
strictly larger the single-qubit stabilizer subtheory.

\textit{Contextuality in the 8-state model.---}We now show that the 8-state
model of the single-qubit stabilizer subtheory exhibits transformation
contextuality, a feature missed in previous studies of this
model~\cite{Wallman2012c, Blasiak2013, Kocia2017}. The single-qubit stabilizer
subtheory consists of preparations and measurements in the eigenbases of the
single-qubit Pauli matrices $\{X,Y,Z\}$, the group of unitary transformations
that permute the signed single-qubit Pauli matrices (i.e. the single-qubit
Clifford group) and convex combinations of these operations. The single-qubit
stabilizer subtheory has the property that preparing an eigenstate of one Pauli
matrix $P$ with eigenvalue $\eta$ then measuring another Pauli $Q$ results in
the eigenvalue $\eta'=\pm \eta$ if $P=\pm Q$ and otherwise results in either
eigenvalue with equal probability.

The 8-state model, originally developed in Ref.~\cite{Wallman2012c}, is a
natural ontological model for the single-qubit stabilizer subtheory (see
\cref{fig:QubitCube}). It is defined by setting $\Lambda = \{\pm 1\}^{\times 3}$
and writing $\lambda = (x,y,z)$, where $x$, $y$, and $z$ are the eigenvalues of
$X$, $Y$, and $Z$ respectively. These ontic states form the extremal points of
the classical probability simplex for three random binary variables. Preparing
the $\eta$ eigenstate of $X$ corresponds to setting $x=\eta$ and choosing $y$
and $z$ uniformly at random, \textit{etc}. Similarly, measuring $X$ returns the
value of $x$, \textit{etc}. This model is both preparation and measurement
non-contextual~\cite{Wallman2012c}.

\begin{figure}[t]
 \centering
    \includegraphics[width=0.5\textwidth]{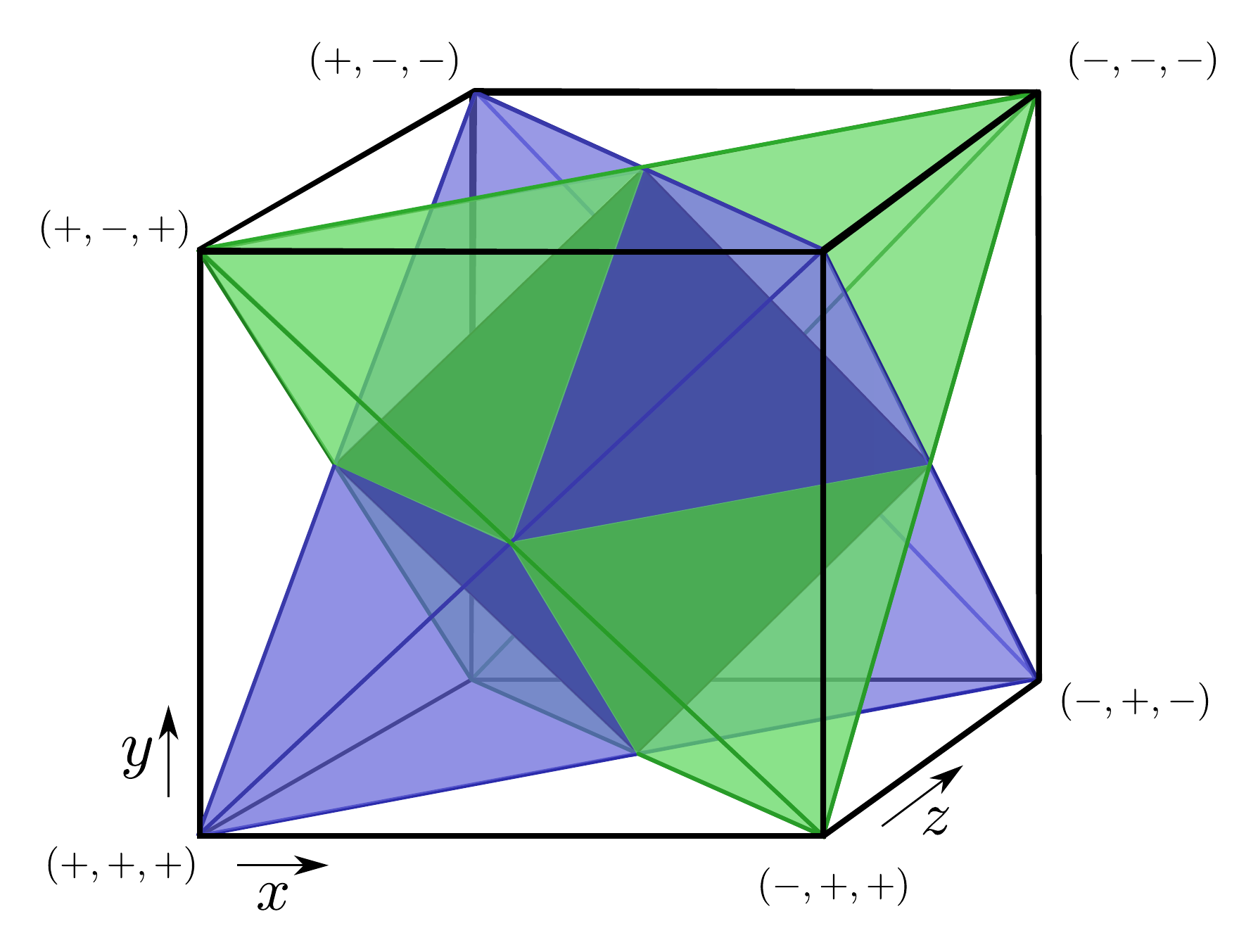}
\caption{A graphical representation of the ontic space $\Lambda$ of the 8-state
model, where the tuples $(x,y,z)$ index ontic states. The green and blue
tetrahedra are the simplices of the odd- and even-parity ontic states
respectively. The stabilizer polytope is the octahedron defined by the
intersection of the two tetrahedra. The transformation $\Gamma_{\cT_1}$ from
\cref{eq:T1} maps any ontic state in one of the tetrahedra to another ontic
state in the same tetrahedra, while the operationally equivalent transformation
$\Gamma_{\cT_2}$ from \cref{eq:T2} maps any ontic state to an ontic state in the
opposite tetrahedra.}
\label{fig:QubitCube}
\end{figure}

In the 8-state model, a transformation corresponds to a permutation that acts on
the hidden variable $(x,y,z)$ in the same way that it acts on the Pauli
operators $(X,Y,Z)$. For example, conjugation by $X$ maps $(X,Y,Z) \to (X, -Y,
-Z)$, and so is represented by the permutation $\Gamma_X: (x,y,z) \to (x, -y,
-z)$, with the transformations for $Y$ and $Z$ defined in a similar manner.
Conjugation by the Hadamard matrix,
\begin{align*}
H = \frac{1}{\sqrt{2}}\begin{pmatrix}
1 & 1 \\
1 & -1
\end{pmatrix},
\end{align*}
maps $(X,Y,Z)\to (Z,-Y,X)$ and so is represented by the permutation $\Gamma_H:(x,y,z)\to (z,-y,x)$. 
Note that a Pauli operation preserves the parity $xyz$ and the Hadamard swaps it (as does the phase gate $P$).

We now show the 8-state model is transformation contextual.
Let
\begin{align}
\cT_1(\rho) &= [\rho + X\rho X + Y\rho Y + Z\rho Z]/4 \notag\\,
\cT_2(\rho) &= H\cT_1(\rho) H.
\end{align}
These two transformations are operationally equivalent, as $\cT_1(\rho) = \cT_2(\rho) = I/2$ for any input state $\rho$.
However, by convexity we have
\begin{align}\label{eq:T1}
\Gamma_{\cT_1}[(a,b,c), (x,y,z)] = \begin{cases}
\frac{1}{4} & \mbox{if } xyz = abc \\
0 & \mbox{otherwise,}
\end{cases}
\end{align}
while, since the Hadamard swaps the sign of $xyz$,
\begin{align}\label{eq:T2}
\Gamma_{\cT_2}[(a,b,c), (x,y,z)] = \begin{cases}
0 & \mbox{if } xyz = abc \\
\frac{1}{4} & \mbox{otherwise.}
\end{cases}
\end{align}
That is, $\Gamma_{\cT_1} \neq \Gamma_{\cT_2}$, as illustrated in
\cref{fig:QubitCube}.

\paragraph*{The single-qubit stabilizer subtheoy is contextual.\textbf{---}}

Above we demonstrated that the 8-state model for the single-qubit stabilizer
subtheory is transformation contextual. We now prove that there is no 
generalized non-contextual model for the single-qubit stabilizer subtheory, and
hence that the single-qubit stabilizer subtheory is contextual. The proof
follows by reducing the ontic space of a general preparation non-contextual
model of the single-qubit stabilizer subtheory to that of the 8-state model.

\begin{theorem}\label{mainthm}
Every ontological model of the single-qubit stabilizer subtheory is either preparation or transformation contextual.
\end{theorem}

\begin{proof}
Fix an arbitrary preparation non-contextual ontological model of the single-qubit stabilizer subtheory.
Let $\Delta_\rho$ be the support of the quantum state $\rho$ in the ontological model, that is the set of physical states $\rho$ has some possibility of preparing,
\begin{align}
\Delta_\rho = \set{\lambda | \, \mu_\rho(\lambda)>0, \lambda\in\Lambda}.
\end{align}
Deleting any ontic state $\lambda\in\Lambda$ such that $\mu_{I/2}(\lambda)=0$,
we can partition $\Lambda$ into 8 disjoint spanning sets from the assumption of
PNC ~\cite[eqs. (11) and (83)--(87)]{Spekkens2005},
\begin{align}
\Lambda_{x,y,z} = \Delta_{(I+xX)/2} \cap \Delta_{(I+yY)/2} \cap \Delta_{(I+zZ)/2}.
\end{align}

As the model is preparation non-contextual, every quantum state has a unique
support. Hence this partitioning is unique. 

Noting that preparing $\sigma$ and then applying a transformation $T$,
that implements a CPTP map $\Phi$, is a valid preparation procedure for the state
$\Phi(\sigma)$. It must be the case that $T$ maps the support of $\rho$ to the support of $\Phi(\rho)$ in a preparation non-contextual ontological model;
\begin{align}\label{eq:MapSupport}
\Gamma_T: \Delta_\rho \to \Delta_{\Phi(\rho)}.
\end{align}
Therefore a Pauli $X$ unitary must be represented by the permutation 
$\tau_X:\Lambda_{x,y,z}\to\Lambda_{x,-y,-z}$ on the partition $\set{\Lambda_{x,y,z}}$. Similarly, Pauli $I$, $Y$, and $Z$
transformations must be represented by the respective permutations
$\tau_I:\Lambda_{x,y,z}\to\Lambda_{x,y,z}$,
$\tau_Y:\Lambda_{x,y,z}\to\Lambda_{-x,y,-z}$, and
$\tau_Z:\Lambda_{x,y,z}\to\Lambda_{-x,-y,z}$. Therefore by convex linearity
there exists an implementation of $\cT_1$ that has the same stochastic map as \cref{eq:T1}, when defined over the
coarse-grained sets $\Lambda_{x,y,z}$.

Similarly for the Hadamard gate we have the map $\tau_H: \Lambda_{x,y,z} \to
\Lambda_{x,-y,z}$. Therefore there exists an implementation of $\cT_2$
that has the same stochastic map as \cref{eq:T2}, when defined over the coarse-grained sets $\Lambda_{x,y,z}$. That is
$\cT_1\cong\cT_2$ and yet they cannot be represented by the same stochastic map
in any preparation non-contextual model.

\end{proof}

We now show that any model of the single-qubit stabilizer subtheory must be either Kochen-Specker
contextual ~\cite{Kochen1967} or transformation contextual.

\begin{theorem}\label{ksthm}
Every ontological model of the single-qubit stabilizer subtheory is either
Kochen-Specker contextual or transformation contextual.
\end{theorem}

\begin{proof}
The proof proceeds in the same manner as \cref{mainthm}, where we
use the fact that Kochen-Specker non-contextuality is implied by the conjunction 
of outcome determinism and generalized measurement
non-contextuality~\cite{Spekkens2005}. By outcome determinism, we can partition
$\Lambda$ into 8 disjoint sets according to the measurement outcomes;
\begin{align*}
\widetilde{\Lambda}_{xyz} = \set{\lambda| \xi_x^X (\lambda) = 1, \, \xi_y^Y (\lambda) = 1, \,\xi_z^Z (\lambda) = 1 },
\end{align*}
By measurement non-contextuality, this partitioning is unique. Using the 
equivalent to \cref{eq:MapSupport} for measurements, the maps $\Gamma_{\cT_1}$ 
and $\Gamma_{\cT_2}$ must be represented as stated in \cref{eq:T1,eq:T2}.
\end{proof}

Note \cref{ksthm} can be used to prove \cref{mainthm}, as preparation non-contextuality implies traditional contextuality for sharp measurements ~\cite{Leifer2013}.

\textit{Discussion.---}In this Letter we have shown that the single-qubit
stabilizer subtheory, a very simple subtheory of the smallest quantum system,
exhibits generalized contextuality. This demonstrates that generalized
contextuality is so prevalent that even an essentially trivial quantum subtheory is classified as contextual, and therefore non-classical. 
This contextuality is only apparent if all single-qubit stabilizer operations are accounted for, that is all stabilizer states, all stabilizer measurements, and the full Clifford group. Therefore a universally non-contextual model can only be constructed for strict subtheories of the singe-qubit stabilizer subtheory. For example, the Hadamard and Phase gates are not elements of the toy theory ~\cite{Spekkens2007} or the standard Wigner function ~\cite{Gibbons2004}, conversely the Hadamard gate is an element of the rebit subtheory ~\cite{Delfosse2015}, but $Y$ eigenstates and $Y$ measurements are not.

Our result also demonstrates that the operational reduction of only considering preparations and measurements is less robust than
previously recognized. As this reduction can conceal key features of the model, 
such as the presence of some forms of contextuality. It is an interesting open
problem to understand how and when this kind of reduction can obscure such
important conceptual features of an operational theory. 
A possible route to investigate this is the Choi-Jamiołkowski isomorphism. The isomorphism, in quantum theory, relates transformations to states in a larger Hilbert space. Hence by using a similar isomorphism for ontological models we may be able to find a connection between the impossibility of a  universal non-contextual ontological model of the single qubit stabilizer subtheory and the impossibility of a preparation non-contextual ontological model of the two-qubit stabilizer subtheory ~\cite{Mermin1993,Leifer2013}.

\begin{acknowledgments}
We would like to thank Rob Spekkens, Stephen Bartlett, Angela Karanjai, and Hammam Qassim for useful discussions. This research was supported by the Government of Canada through CFREF Transformative Quantum Technologies program, NSERC Discovery program, and Industry Canada.
\end{acknowledgments}

\bibliography{UnderlyingOntology}

\appendix

\section{Appendix: Measure-Theoretic Treatment of Theorem 1}
\label{app:measure}

We now prove \cref{mainthm} in the more general measure-theoretic framework for
ontological models, see ~\cite{Leifer2014} definition 8.2.

\begin{theorem}
Every ontological model of the single-qubit stabilizer subtheory is either 
preparation or transformation contextual.
\end{theorem}

\begin{proof}
To adapt the proof to the measure-theoretic setting, we need to change the 
definition of support of a quantum state as follows. Let $\Delta_\rho$ be the 
support of the quantum state $\rho$ in the ontological model, that is, a (not
necessarily unique) set such that for all $S\in\Sigma$,
\begin{align}\label{eq:MSupp}
\mu_\rho(S) 
\begin{cases}
=1 & \mbox{if } \Delta_\rho\subseteq S\\
<1 & \mbox{otherwise.}
\end{cases}
\end{align}
As before we delete any measurable set $S\in\Sigma$ such that $\mu_{I/2}(S)=0$,
then partition $\Lambda$ into 8 spanning sets that intersect on sets of measure
zero, from the assumption of PNC ~\cite[eqs. (11) and (83)--(87)]{Spekkens2005}.
Having reduced the model to a model over a finite set of states, the rest of the 
proof follows as described in the main text.
\end{proof}

\end{document}